\documentclass[runningheads,a4paper]{llncs}
\usepackage{amsmath}
\usepackage{amssymb}
\usepackage{graphicx}
\usepackage{prooftree}
\usepackage{url}
\usepackage{stmaryrd}
\usepackage{lingmacros}
\urldef{\mailsa}\path|{alfred.hofmann, ursula.barth, ingrid.haas, frank.holzwarth,|
\urldef{\mailsb}\path|anna.kramer, leonie.kunz, christine.reiss, nicole.sator,|
\urldef{\mailsc}\path|erika.siebert-cole, peter.strasser, lncs}@springer.com|    
\newcommand{\disp}[1]{{\footnotesize\enumsentence{#1}}}

\newcommand{\techterm}[1]{{\it #1}}

\newcommand{\mycomment}[1]{} 
\newcommand{\commentout}[1]{}
\newcommand{\product}{\mbox{$\bullet$}}
\newcommand{\ninfix}{\mbox{$\Downarrow$}}
\newcommand{\nextract}{\mbox{${\Uparrow}$}}
\newcommand{\bsl}{\mbox{$\backslash$}}
\newcommand{\yields}{\mbox{\ $\Rightarrow$\ }}

\newcommand{\tb}{\hspace*{0.25in}}

\newcommand{\oscott}{\mbox{$\llbracket$}}
\newcommand{\cscott}{\mbox{$\rrbracket$}}

\newcommand{\mysplit}{\mbox{\v{}}}

\newcommand{\splitk}[1]{\mbox{\v{}$^{_{#1}}$}}

\newcommand{\infix}{\mbox{$\downarrow$}}
\newcommand{\extract}{\mbox{${\uparrow}$}}
\newcommand{\dprod}{\mbox{$\odot$}}
\newcommand{\oldseg}[2]{\mbox{$\sqrt[#1]{#2}$}}% type segment, old version
\newcommand{\omg}{\mbox{$\omega$}} 
\newcommand{\powerset}[1]{\mbox{$\mathcal P(#1)$}}
\newcommand{\reduces}{\mbox{$\,\rhd\,$}}
\newcommand{\reducest}{\mbox{$\,\rhd\!^*\,$}}

\newcommand{\bydef}{\mbox{$\;=\;$}} 
\newcommand{\sep}{\mbox{$1$}}

\newcommand{\vect}[1]{\overrightarrow{#1}}% vectorial notation

\newcommand{\rightproj}{\mbox{$\triangleright^{-1}$}}
\newcommand{\leftproj}{\mbox{$\triangleleft^{-1}$}}

\newcommand{\vtab}{\ \vspace{2.5ex}\\}

\newcommand{\LC}{\mbox{$\mathbf{L1}$}}
\newcommand{\LCstar}{\mbox{$\mathbf{L1}$}}
\newcommand{\D}{\mbox{$\mathbf{D}$}}
\newcommand{\Dstar}{\mbox{$\mathbf{D}$}}
\newcommand{\Dstarimp}{\mbox{$\mathbf{D}[\rightarrow]$}}
\newcommand{\da}{\mbox{$\mathcal{DA}$}}

\newcommand{\rd}{\mbox{$\mathcal{RD}$}}

\newcommand{\prsd}{\mbox{$\mathcal{PRSD}$}}
\newcommand{\prdd}{\mbox{$\mathcal{PRDD}$}}
\newcommand{\invprov}[1]{\mbox{$[#1]$}}% In the style of Okada's papers on sem Cut Elim: [A]\bydef\{\Delta\mid cD\vdash \Delta\yields A\}
\newcommand{\rep}[1]{\mbox{$\overline{#1}$}} %representative of class of equivalence

\newcommand{\comp}[1]{\mbox{$\circ_{#1} $}}

\newcommand{\interp}[3]{\mbox{$\oscott #1\cscott_{#3}^{\cal{#2}}$}}% interpretation
\newcommand{\alg}[1]{\mbox{$\mathcal{#1}$}}% algebra
\newcommand{\car}[1]{\mbox{$|\alg{#1}|$}}% carrier of an algebra
\newcommand{\model}[1]{\mbox{$\mathcal{#1}$}}% model
\newcommand{\sigmad}{\mbox{$\Sigma_D$}}% A \Sigma_D algebras, the signature of DAs
\newcommand{\myunder}{\mbox{$\bsl\!\bsl$}}% powerset operation
\newcommand{\myover}{\mbox{$/\!/$}}% powerset operation
\newcommand{\myextract}{\mbox{$\upuparrows$}}% powerset operation
\newcommand{\myinfix}{\mbox{$\downdownarrows$}}% powerset operation
\newcommand{\myconunit}{\mbox{$\mathbb{I}$}}% powerset operation
\newcommand{\mydiscunit}{\mbox{$\mathbb{J}$}}% powerset operation
\newcommand{\cD}{\mbox{$\mathbf{cD}$}}% categorical calculus
\newcommand{\hD}{\mbox{$\mathbf{hD}$}}% hypersequent calculus
\newcommand{\tp}{\mbox{$\mathbf{Tp}$}}
\newcommand{\tpstar}{\mbox{$\mathbf{Tp}$}}
\newcommand{\pr}{\mbox{$\mathbf{Pr}$}}
\newcommand{\config}{\mbox{$\mathcal O$}}%{\mbox{$\mathbf{HConfig}$}}
\newcommand{\configstar}{\mbox{$\mathcal O$}}%{\mbox{$\mathbf{HConfig}_*$}}
\newcommand{\segtyp}{\mbox{$\mathbf{Seg}$}}
\newcommand{\segtypstar}{\mbox{$\mathbf{seg}$}}
\newcommand{\subst}[2]{\mbox{$\left(\begin{array}{c}#1\\#2\\\end{array}\right)$}}
\newcommand{\myplus}{\mbox{$\!+\!$}}
\newcommand{\mytimes}{\mbox{$\!\times\!$}}

\newcommand{\Irr}{\mbox{$\mathbf{Irr}$}}
\newcommand{\T}{\mbox{$\mathcal{T}$}}%(\mathbf{Init})

\newcommand{\timesnd}[1]{\mbox{$\tilde{\times}_{#1}$}}
\newcommand{\f}{\mbox{$\mathbf{f}$}}
\newcommand{\s}{\mbox{$\mathbf{s}$}}
\newcommand{\nDstar}{$\mathbf{nD}_*$}
\newcommand{\ndyields}{\mbox{$\twoheadrightarrow$}}% \yields for \nDstar
\newcommand{\equi}{\mbox{$\;\leftrightarrow\;$}}

\mainmatter  
\title{Models for the Displacement Calculus}

\author{Oriol Valent\'in%
\thanks{Research partially supported by
SGR2014-890 (MACDA) of the Generalitat de Catalunya,
and MINECO project APCOM (TIN2014-57226-P).}% XXX
}
\institute{Universitat Polit\`ecnica de Catalunya\\
}

\toctitle{Lecture Notes in Computer Science}
\tocauthor{Authors' Instructions}
\begin{document}
\maketitle
\begin{abstract}
The displacement calculus \D{} is a conservative extension of the Lambek calculus
\LCstar{} (with empty antecedents allowed in sequents). \LCstar{} can be said to be the logic of
concatenation, while \D{} can be said to be the logic of concatenation and intercalation. In many senses,
it can be claimed that \D{} mimics \LCstar{} in that the proof theory, generative capacity and complexity
of the former calculus are natural extensions of the latter calculus. In this paper, we strengthen this claim.
We present the appropriate classes of models for \D{} and prove some completeness results; strikingly, we
see that these results and proofs are natural extensions of the corresponding ones for \LCstar{}. 
\end{abstract}
\section{Introduction}\label{sectintro}
The displacement calculus \D{} is a quite well-studied extension of the Lambek calculus
\LCstar{} (with empty antecedents allowed in sequents). In many papers (see \cite{mv:disp}, \cite{mvf:tdc} and \cite{mvf:tbilisi}),
\D{} has proved to provide elegant accounts of a variety of linguistic phenomena of English, and of Dutch, namely a processing
interpretation of the so-called Dutch cross-serial dependencies.

The hypersequent format \hD{}\footnote{Not to be confused with the hypersequents of Avron (\cite{avron:91}).} of
displacement calculus is a pure
sequent calculus free of structural rules which subsumes the sequent calculus for \LCstar{}. The Cut elimination algorithm for
\hD{} provided in \cite{mvf:tdc} mimics the one of Lambek's \cite{lambek:mathematics} syntactic calculus (with some minor differences concerning the possibility
of empty antecedents). Like \LCstar{}, \D{} enjoys some nice properties such as the subformula property, decidablity, the finite
reading property and the focalisation property (\cite{mv:cl}).

Like \LCstar{}, \D{} is known to be NP-complete {\cite{moot:tdc}}. Concerning (weak) generative capacity, \D{} recognises the class of well-nested multiple context-free languages (\cite{sorokin:gencapD}). In this respect, the result on generative
capacity generalises the result that states that \LCstar{} recognises the class of context-free languages. One point of divergence in terms of generative capacity is that \D{} recognises the class of the permutation closures of context-free languages (\cite{mv:tag+10}). Finally,
it is important to note that a Pentus-like upper bound theorem for \D{} is not known.

In this paper we present natural classes of models for \D{}. Several strong completeness results are proved, in particular strong completeness w.r.t.\ the class of residuated displacement algebras (a natural extension of residuated monoids). Powerset frames
for \LCstar{} are of interest from the linguistic point of view because of their
relation to language models. Powerset residuated displacement algebras over displacement
algebras are given, which generalise the powerset residuated monoids over monoids, as well as 
over free monoids. Strong completeness
results for the so-called implicative fragment of \D{}, which is very relevant linguistically, is proved in the spirit
of Buszkowski (\cite{buszkowski:comp86}), but the construction is more subtle.
%Moreover, full completeness with respect powerset residuated displacement algebras over displacement algebras is given following again Buszkowski's methods.

The structure of the paper is as follows. In Section~\ref{sectdda} we present the basic proof-theoretic tools (useful for the construction of canonical models) which we shall employ for the study of \D{}
from a semantic point of view. In Section~\ref{sectstrongcompldstarimpl} we provide the proof of two strong completeness of what we call the \emph{implicative\/} fragment w.r.t.\ powerset DAs over standard DAs (with a countably infinite set
of generators) and L-models respectively. 
%In Section~\ref{fullDcomplsect}, the proof of strong completeness of full \D{} w.r.t.\ powerset residuated displacement %algebras over displacement algebras is outlined. 
%Finally, we conclude in the last section.
% Conventions on notation
\section{The Categorical calculus \cD{} and the Hypersequent Calculus \hD{} }\label{sectdda}
% Insertion from lambek_hidden
\D{} is model-theoretically motivated, and the key to its conception is the use of many-sorted universal algebra
(\cite{goguen}), namely $\omg$-sorted universal algebra. Here, we assume a version of many-sorted algebra such that the sort domains of an $\omg$-sorted algebra $\alg A$ are non-empty. With this condition we avoid some pathologies which arise in a na\"ive version of many-sorted universal algebra (cf.\ \cite{goguen}, and \cite{lalement}).
Some definitions are needed. Let $\alg{M}=(\car{M}, +,$ $ 0,\sep)$ be a free monoid
where $\sep$ is a distinguished element of the set of generators $X$ of $\alg M$. We call such an algebra a \emph{separated monoid}.
Given an element $a\in\car{M}$, we can associate to it a number, called its \emph{sort} as follows:
\disp{
$
\begin{array}[t]{lll}
s(\sep)&=&1\\
s(a)&=&0\mbox{ if $a\in X$ and $a\neq \sep$}\\
s(w_1+w_2)&=&s(w_1)+s(w_2)
\end{array}
$
\label{sortdesda}
}
This induction is well-defined since $\alg{M}$ is free and \sep{} is a (distinguished) generator;
the sort function $s(\cdot)$ in a separated monoid simply counts the number of separators an element contains.
\begin{definition}\rm({\textbf{Sort Domains}})\\
Where $\alg{M}=(\car{M}, +, 0,\sep)$
is a separated monoid,
the {\em sort domains\/} $\car{M}_i$ of sort $i$ are defined as follows:
$$
\begin{array}{rcl}
\car{M}_i & = & \{a\in\car{M}: s(a) = i\}, i\geq 0
\end{array}
$$
%\label{defsepmon}
It is readily seen that for every $i,j\geq 0$, $\car{M}_i\cap\car{M}_j=\emptyset\mbox{ iff }i\neq j$.
\end{definition}
% standard DA
\begin{definition} \rm({\textbf{Standard Displacement Algebra}})\\
The \emph{standard displacement algebra} (or standard DA) defined by a separated monoid $(\car M, +, 0, \sep)$
is the $\omg$-sorted algebra with the $\omg$-sorted signature $\Sigma_D=(+,\{\times_{i}\}_{i>0},0,1)$ 
with sort functionality $((i,j\rightarrow i+j)_{i,j\geq 0},(i,j\rightarrow i+j-1)_{i>0,j\geq 0},0,1)$:
$$(\{\car{M}_i\}_{i\geq 0},+,
\{\times_{i}\}_{i>0}, 0, 1)$$
where:
$$
\begin{array}{|l||l|}
\hline
\mbox{operation} & \mbox{which is}\\
\hline\hline
+: \car{M}_i\times \car{M}_j\rightarrow \car{M}_{i+j} & \begin{minipage}{43ex}as in the separated
monoid\end{minipage}\\
\hline
\times_k: \car{M}_{i}\times \car{M}_j\rightarrow \car{M}_{i+j-1} & \begin{minipage}{43ex}$\times_k(s, t)$: the result of
replacing the $k$-th separator in $s$ by $t$\end{minipage}\\
\hline
\hline
\end{array}$$
\label{chapt4opsDA}
\end{definition}
%We will usually write standard DA instead of standard displacement algebra.
% Insertion recursive definition of Tp
\noindent The sorted types of $\mathbf{D}$, which we will interpret residuating w.r.t\ the sorted operations in 
Definition~\ref{chapt4opsDA}, are defined by mutual recursion in Figure~\ref{chapt4typesD}. We let 
$\tp=\bigcup_{i\geq 0}\tp_i$. A subset $B$ of $\car{M}$ is called a \emph{same-sort} subset iff there exists an $i\in\omg$ such that for every $a\in B$,
$s(a)=i$. \D{} types are to be interpreted as same-sort subsets of $\car{M}$.
% Mutual recursion definition of Tp
\begin{figure}
$$\begin{array}{rclll}
\tp_i & ::= & \pr_i &\mbox{where }\pr_i\mbox{ is the set of atomic types of sort }i\\\\
\tp_{0}& ::= & I&\mbox{Continuous unit}\\
\tp_{1}& ::= & J&\mbox{Discontinuous unit}\\\\ 
\tp_{i+j} & ::= & \tp_i\product\tp_j&\mbox{continuous product}\\
\tp_j & ::= & \tp_i\bsl\tp_{i+j}&\mbox{under}\\
\tp_i & ::= & \tp_{i+j}/\tp_j&\mbox{over}\\\\
\tp_{i+j} & ::= & \tp_{i+1}\dprod_k\tp_j&\mbox{discontinuous product}\\
\tp_j & ::= & \tp_{i+1}\infix_k\tp_{i+j}&\mbox{extract}\\
\tp_{i+1} & ::= & \tp_{i+j}\extract_k\tp_j&\mbox{infix}
 \end{array}$$
\caption{The sorted types of $\mathbf{D}$}
\label{chapt4typesD}
\end{figure}
% Standard type interpretation
\noindent I.e.\ every inhabitant of $\oscott A\cscott$ has the same sort. The intuitive semantic interpretation of the connectives is shown in Figure~\ref{typeint}; this interpretation is called the \techterm{standard interpretation}.
Observe that for any type $A\in\tp$, the interpretation of $A$, i.e.\ $\oscott A\cscott$,
is contained in $M_{s(A)}$, where the sort map $s(\cdot)$ for the set \tp{}, is such that
\disp{
$
\begin{array}[t]{lllllllll}
s( p ) &=&i&\mbox{ for p}\in\pr_i&\\
s(I)&=&0\\
s(J)&=&1\\
s(A\product B)&=&s(A)+s(B)\\
s(A\bsl B)&=&s(B)-s(A)\\
s(B/A)&=&s(B)-s(A)\\
s(A\odot_k B)&=&s(A)+s(B)-1\\
s(A\infix_k B)&=&s(B)-s(A)+1\\
s(B\extract_k A)&=&s(B)-s(A)+1
\end{array}
$
\label{sortmapD}
} 
% Standard semantic interpretation of \D{} types
\begin{figure}
\begin{small}
$
\begin{array}[t]{rcll}
%\oscott p \cscott & \subseteq &\car{M}_i\mbox{ for } i\geq 0 &p\in\pr_i\\\\
\oscott I\cscott & = &\{0\}&\mbox{continuous unit}\\
\oscott J\cscott & = &\{1\}&\mbox{discontinuous unit}\\\\
\oscott A\product B\cscott & = &
\{s_1+ s_2|\ s_1\in\oscott A\cscott\ \&\
s_2\in\oscott B\cscott\} & \mbox{product}\\
\oscott A\bsl C\cscott & = & \{s_2|\ \forall s_1\in\oscott A\cscott,
s_1+ s_2\in\oscott C\cscott\} & \mbox{under}\\
\oscott C/B\cscott & = & \{s_1|\ \forall s_2\in\oscott B\cscott,
s_1+ s_2\in\oscott C\cscott\} & \mbox{over}\\\\
\oscott A\dprod_k B\cscott & = &
\{\times_k(s_1, s_2)|\ s_1\in\oscott A\cscott\ \&\
s_2\in\oscott B\cscott\} & k>0
 \mbox{ \normalfont{discontinuous product}}\\
\oscott A\infix_k C\cscott & = & \{s_2|\ \forall s_1\in\oscott A\cscott,
\times_k(s_1, s_2)\in\oscott C\cscott\} & k>0\mbox{ \normalfont{infix}}\\
\oscott C\extract_k B\cscott & = & \{s_1|\ \forall s_2\in\oscott B\cscott,
\times_k(s_1, s_2)\in\oscott C\cscott\} & k>0\mbox{ \normalfont{extract}}\\
\\\\
\end{array}
$
\end{small}
\caption{Standard semantic interpretation of $\mathbf{D}$ types}
\label{typeint}
\end{figure}
\subsection{\D{} and its Categorical Presentation \cD{}}\label{subsectcD}
In~\cite{valentin:phd} \D{} is presented as a categorical calculus:
%it is proved that the identities (or equations) true of standard DAs has as equational theory the so-called class of %(general) displacement algebras (DA) (see Figure~\ref{eqd2}).
% The faithful embedding 
% Insertion of Eq_2
% Eq2 axiomatization
\begin{figure}[h]
\begin{flushleft}
$
\begin{array}{l}
%\hline
\\
\mbox{\textbf{Continuous associativity }}\\
x+(y+ z)\approx (x+ y)+ z
\\\\

\mbox{\textbf{Discontinuous associativity}}\\
x\times_i(y\times_j z)\approx (x\times_i y)\times_{i+j-1} z
\mbox{ }\\
(x\times_i y)\times_j z\approx x\times_i(y\times_{j-i+1} z)\mbox{
 if }i\leq j\leq 1+s(y)-1\\\\

\mbox{\textbf{Mixed permutation}} \\
(x\times_i y)\times_j z\approx (x\times_{j-\,S(y)\,+1}z)\times_{i} y\mbox{ if }j>i+s(y)-1\\
(x\times_{i} z)\times_{j} y
\approx (x\times_{j} y)\times_{i+S(y)-1}z \mbox{ if }j<i\\\\

\mbox{\textbf{Mixed associativity}}\\
(x+ y)\times_{i} z \approx (x\times_{i}z)+ y\mbox{ if }1\leq i\leq s(x)\\
(x+ y)\times_{i} z \approx x+ (y\times_{i-s(x)}z)\mbox{ if }x+1\leq i\leq s(x)+s(y)\\\\

\mbox{\textbf{Continuous unit and discontinuous unit}}\\
0+ x\approx x\approx x+ 0\mbox{ and }1\times_{1} x \approx x\approx x\times_{i} 1 \\
%\hline
\end{array}
$
\end{flushleft}
\caption{Axiomatisation of \da{}}
\label{eqd2}
\end{figure}
\disp{
$
\begin{array}[t]{lll}
A\rightarrow A\mbox{ Axiom}\\
A\product B\rightarrow C\mbox{ iff }A\rightarrow C/B\mbox{ iff }B\rightarrow A\bsl C\quad\quad\quad\mbox{ $Res_{cont}$}\\
A\odot_i B\rightarrow C\mbox{ iff }A\rightarrow C\extract_i B\mbox{ iff }B\rightarrow A\infix_i C\quad\ \mbox{ $Res_{disc}$}\\\\
A\bullet I\equi A \equi I\bullet A\quad A\odot_i J\;\equi\; A \;\equi\; J\odot_1 A\\
(A\bullet B)\bullet C\;\equi\;A\bullet (B\bullet C)\quad\quad\quad\quad\quad\mbox{ Continuous associativity}\\
A\odot_i(B\odot_j C)\equi (A\odot_i B)\odot_{i+j-1} C\quad\mbox{ Discontinuous associativity}\\
(A\odot_i B)\odot_j C\equi A\odot_i(B\odot_{j-i+1} C)
\mbox{, if }i\leq j\leq 1+s(B)-1\\\\
(A\odot_i B)\odot_j C\equi (A\odot_{j-s(B)+1} C)\odot_{i}B, \mbox{ if }j> i +s(B)-1\mbox{ Mixed permutation}\\
(A\odot_i C)\odot_j B\equi (A\odot_j B)\odot_{i+s(B)-1}C,\mbox{ if }j < i\\\\
(A\bullet B)\odot_i C\equi (A\odot_i C)\bullet B\mbox{, if }1\leq i\leq S(A)\mbox{ Mixed associativity}\\
(A\bullet B)\odot_i C\equi A\bullet(B\odot_{i-s(C)}C)\mbox{, if }s(A)+1\leq i\leq s(A)+s(B)\\
\mbox{From }A\rightarrow B\mbox{ and } B\rightarrow C \mbox{ we have }A\rightarrow C\quad\mbox{ Transitivity}
\end{array}
$
\label{catcalcdef}
}
In Figure~\ref{eqd2} we find the axiomatisation of the class of DAs \da{}. 
Just as in the case of \LC{}, the natural class of algebras is the class of residuated monoids $\cal RM$, in the case
of \D{}, the natural class of algebras is the class of residuated displacement algebras (residuated DAs) $\rd$.

One can restrict the definition of the sorted types. Let $\mathbf C$ be a subset of the connectives considered in the definition of types in Figure~\ref{chapt4typesD}. We define $\tp[\mathbf C]$ as the least set of sorted types generated by 
\pr{} and the set of connectives $\mathbf C$. If the context is clear, we will write \tp{} instead of $\tp[\mathbf C]$.

Let us define the formal definition of a model. A model $\mathcal M=(\alg{A}, v)$ comprises a 
(residuated) \sigmad{}-algebra
and a $\omg$-sorted mapping $v:\pr\rightarrow \tp[\mathbf C]$ called a valuation. The mapping $\widehat v$ is 
the unique function which extends $v$ and which is such that $\widehat v(A*B)=\widehat v(A)*\widehat v(B)$ (if * is a binary
connective of $\mathbf C$) and $\widehat v(*A)=* \widehat v(A)$ (if * is a unary connective of $\mathbf C$). Finally, a $0$-ary connective is mapped into the corresponding unit of $\car A$. Needless to say, the mappings $v$ and $\widehat v$ preserve the sorting regime. 
%connective is mapped into an element of $\car A$,
% Definition

Let us see that $\D$ (with all the connectives) is strongly complete w.r.t. \rd{}. Soundness is trivial because
we are considering the categorical calculus $\cD$.
For completeness,
we can define the well-known Lindenbaum-Tarski construction to see that \cD{} is strongly complete w.r.t. \rd{}. The canonical model is $(\alg{L},v)$ where $\alg{L}$ is $(\tp/\theta,\comp{}, (\comp{i})_{i>0},\myunder,\myover,(\myinfix_i)_{i>0}(\myextract_i)_{i>0},\mathbb I,\mathbb J;\leq)$ where the interpretation of the new symbols is as expected. Let $\theta_R$ be the equivalence relation on \tp{} defined as follows: $A\theta_R B$ iff $R\vdash_{\cD}
A\rightarrow B$ and $R\vdash_{\cD}B\rightarrow A$, where $R$ is a set of non-logical axioms. Using the usual tonicity properties for the connectives of \tp{}, one proves that $\theta_R$ is a congruence. Where $A$ is a type, $\rep A$ is an element of $\tp/\theta_R$, i.e.\ \tp{} modulo $\theta_R$. We define $\rep A\leq \rep B$ iff $R\vdash_{\cD}A\rightarrow B$. We define the valuation $v$ as $v( p )=\rep p$ ($p$ is a primitive type). We have that for every type $A$, $\widehat v(A)=\rep A$. 
 Finally, one has that $(\alg L,v)\models A\rightarrow B$ iff 
$R\vdash_{\cD}A\rightarrow B$. From this, we infer the following theorem:
\begin{theorem}
The calculus \cD{} is strongly complete w.r.t.\ \rd{}.
\label{strongcomplRD}
\end{theorem}
% The equational theory of \da{} 
% The class \rd{}
% Other DAs like powerset DAs over DAs

Since $\da$ is a variety\footnote{The term \emph{equational class\/} is sometimes used in the literature.} (see Figure~\ref{eqd2}), it is closed by subalgebras, direct products and homomorphic images, which give additional
DAs.

We have other interesting examples of DAs, for instance the \emph{powerset DA over }$\alg{A}=
(\car{A},+,\{\times_i\}_{i>0},0,1)$, which we
denote $\powerset{A}$. We have:
\disp{
$
\powerset{A}=(\car{\powerset{A}},\cdot,\{\comp{i}\}_{i>0},\mathbb I,\mathbb J)
$
\label{defpowsetda}
}
%We recall the definition of \emph{same-sort\/} subsets. A subset $B$ of $\car{A}$ is called a \emph{same-sort} subset iff:
%\disp{
%$
%\mbox{There exists an }i\in\omg\mbox{ such that for every }a\in B,\mbox{ }s(a)=i
%$
%\label{samesordef} 
%}
The notation of the carrier set of $\powerset{A}$ presupposes that its members are same-sort subsets;
notice that $\emptyset$ vacuously satisfies the \emph{same-sort\/} condition.
Where $A$, $B$ and $C$ denote same-sort subsets of $\car{A}$, the operations $\mathbb{I}$, $\mathbb{J}$, 
$\cdot$ and $\comp{i}$ are defined as follows:
\disp{
$
\begin{array}[t]{lll}
\mathbb{I}&=&\{0\}\\
\mathbb{J}&=&\{1\}\\
A\cdot B&=&\{a+b:a\in A\mbox{ and }b\in B\}\\
A\comp{i}B&=&\{a\times_i b:a\in A\mbox{ and }b\in B\}\mbox{ }
\end{array}
$
\label{defcompcompi}
}
\noindent It is readily seen that for every $\alg A$, $\powerset{A}$ is in fact a DA. Notice that every sort domain
$\car{\powerset{A}}_i$ is a collection of same-sort subsets, that the sort domains
of $\powerset{A}$ are non-empty, but no longer satisfy that $\car{\powerset{A}}_i\cap \car{\powerset{A}}_j=
\emptyset$ iff
$i\neq j$, since the empty set $\emptyset\in\car{\powerset{A}}_i$ for every $i\geq 0$. 
% Residuated powerset DAs
A residuated powerset displacement algebra over a displacement algebra $\powerset A$ is the following:
\disp{
$\powerset{A}=(\car{\powerset{A}},\cdot,\myunder,\myover,\{\comp{i}\}_{i>0},, 
\{\myextract_i\}_{i>0},\{\myinfix_i\}_{i>0},\mathbb I,\mathbb J;\subseteq)$
\label{defpowersetresda}
}
where $\myunder$, $\myover$, $\myextract_i$ and $\myinfix_i$ are defined
as follows:
\disp{
$
\begin{array}[t]{lll}
A\myunder B&=&\{d:\mbox{ for every }a\in A,\mbox{ }a+d\in B\}\\
B\myover A&=&\{d:\mbox{ for every }a\in A,\mbox{ }d+a\in B\}\\
B\myextract_i A&=&\{d:\mbox{ for every }a\in A,\mbox{ }d\times_i a\in B\}\\
A\myinfix_i B&=&\{d:\mbox{ for every }a\in A,\mbox{ }a\times_i d\in B\}\\
\end{array}
$
\label{defmyunderetc}
}
The class of powerset residuated DAs over a DA is denoted \prdd{}. The class of powerset residuated DAs over a standard
DA is denoted \prsd{}. Finally, the subclass of \prsd{} which is formed by powerset residuated algebras over 
finitely-generated standard DA are known simply as \emph{L-models\/}. 

Every standard DA $\alg{A}$ has two remarkable properties, namely the property that sort domains $\car{A}_i$ 
(for $i>0$) can be defined in terms of $\car{A}_0$, and the property that every element $a$ of a sort domain $\car{A}_i$
is decomposed uniquely around the separator \sep{}:
\disp{
$
\begin{array}[t]{ll}
\mbox{(S1) }\mbox{For $i>0$, }\car{A}_i=\underbrace{\car{A}_0\comp{}\{1\}\cdots\{1\}\comp{}\car{A}_0}_{\mbox{$(i-1)$ $1's$}}\\
\mbox{(S2) }\mbox{For $i>0$, if $a_0+\sep+\cdots+\sep +a_i = b_0+\sep+\cdots+\sep +b_i$ then }\\
a_k=b_k\mbox{ for }0\leq k\leq i
\end{array}
$
\label{sepdaprop}
}
%The properties $(S1)$ and $(S2)$
% Strong completeness w.r.t. \rd{}
%

Standard DAs, as their name suggests, are particular cases of (general) DAs:
\begin{lemma}
The class of standard DAs is a subclass of the class of DAs.\footnote{Later we see that the inclusion is proper.}
\label{stdaaredas}
\end{lemma}
\begin{proof}
We define a useful notation which will help us to prove the lemma. Where $\mathcal{A}=(\car{A},+,(\times_i)_{i>0},0,1)$
is a standard DA, let $a$ be an arbitrary element of sort $s(a)$. We associate to every $a\in\car{A}$ a sequence of elements $a_0,\cdots,a_{s(A)}$. We have the following vectorial notation:
\disp{
$
\vect{a}_i^j=\left\{
\begin{array}{l}
a_i\mbox{, if }i=j\\
\vect{a}_i^{j-1}+\sep+a_{j}\mbox{, if }j-i>0
\end{array}
\right.
\label{decompositionelstda}
$}
Since $\alg{A}$ is a standard DA, the $a_i$ associated to a given $\vect a$ are unique (by freeness of the underlying monoid). We have that $a=\vect{a}_0^{s(A)}$, and we write $\vect a$ in place of $\vect{a}_0^{s(A)}$.
Consider arbitrary elements $\vect a$, $\vect b$ and $\vect c$ of $\car A$:
\begin{itemize}
\item[\product{}] Continuous associativity is obvious.\\
\item[\product{}] Discontinuous associativity. Let $i, j$ be such that $i\leq j\leq i+s(\vect a)-1$:
$$
\begin{array}{l}
\vect{b}\!\times_j\!\vect{c}= \vect{b}_0^{i-1}\!+\!\vect{c}\!+\!\vect{b}_i^{s(b)}\mbox{, therefore:}\\
\vect{a}\!\times_i\!(\vect{b}\!\times_j\!\vect{c})= \boxed{\vect{a}_0^{i-1}\!+\!\vect{b}_0^{j-1}\!+\!\vect{c}\!+\!\vect{b}_j^{s(b)}
\!+\!\vect{a}_i^{s(a)}}
\end{array}\quad\quad\hfill{(*)}
$$%\label{lemmasdda1}
On the other hand, we have that:
$$
\begin{array}{l}
\vect{a}\times_i\vect{b}= \vect{a}_0^{i-1}\myplus\vect{b}\myplus\vect{a}_i^{s(\!\small{{\vect a}}\!)}
=\vect{a}_0^{i-1}\myplus\vect{b}_0^{j-1}\myplus\!\underbrace{\sep}_{(i+j-1)\scriptsize{\mbox{-th }separator}}\!\myplus\vect{b}_j^{s(\!\small{{\vect b}}\!)}\myplus\vect{a}_i^{s(\!\small{{\vect a}}\!)}\\
\end{array}
$$
It follows that:
$$(\vect{a}\mytimes_i\vect{b})\mytimes_{i+j-1}\vect{c}= \boxed{\vect{a}_0^{i-1}\myplus\vect{b}_0^{j-1}\myplus\vect{c}
\myplus\vect{b}_j^{s(\!\small{{\vect b}}\!)}
\myplus\vect{a}_i^{s(\!\small{{\vect a}}\!)}}\quad\quad(**)$$\label{lemmasdda2}
By comparing the right hand side of (*) and (**), we have therefore: 
$$\vect{a}\!\times_i\!(\vect{b}\!\times_j\!\vect{c})=(\vect{a}\mytimes_i\vect{b})\mytimes_{i+j-1}\vect{c}$$
% Mixed permutation
\item[$\product$] Mixed Permutation. Consider $(\vect a\mytimes_i \vect b)\mytimes_j \vect c$ and suppose that $i+s(\vect b)-1<j$:
$$\vect{a}\mytimes_i\vect{b} = \underbrace{\vect{a}_0^{i-1}\myplus\vect{b}\myplus\vect{a}_i^{j-s(\vect{b})}}_{
j-s(\vect{b})+s(\vect{b})-1=j-1\mbox{ separators}}\!\!
+ 1
\myplus\vect{a}_{j-s(\vect{b}+1)}^{s(\!\small{{\vect a}}\!)}$$
It follows that:\\
$$(\vect a\mytimes_{i}\vect b)\mytimes_{j} \vect c = \boxed{\vect{a}_0^{i-1}\myplus\vect{b}\myplus\vect{a}_i^{j-s(\vect{b})}\myplus
\vect{c}\myplus\vect{a}_{j-s(\vect{b})+1}^{s(\vect{a})}}\quad\quad(***)$$\label{lemmasdda3}
Since $i+s(\vect{b})-1<j$, then $i<j-s(\vect{b})+1$. Then we have that:
$$\vect{a}\mytimes_{j-s(\vect{b})+1}\vect{c}=\vect{a}_0^{i-1}\myplus 1\myplus\vect{a}_i^{j-s(\vect{b})}\myplus\vect{c}\myplus\vect{a}_{j-s(\vect{b})+1}^{s(\vect{a})}$$
It follows that:\\
$$
(\vect{a}\mytimes_{j-s(\vect{b})+1}\vect{c})\mytimes_i\vect{b}= 
\boxed{\vect{a}_0^{i-1}\myplus\vect{b}\myplus\vect{a}_i^{j-s(\vect{b})}\myplus
\vect{c}\myplus\vect{a}_{j-s(\vect{b})+1}^{s(\vect{a})}}\quad\quad(****)
$$\label{lemmasdda4}
By comparing the right hand side of (***) and (****), we have therefore:
$$(\vect a\mytimes_{i}\vect b)\mytimes_{j} \vect c =(\vect{a}\mytimes_{j-s(\vect{b})+1}\vect{c})\mytimes_i\vect{b}$$
\item[$\product$]Mixed associativity. There are two cases: $i\leq s(\vect{a})$ or $i>s(\vect{a})$.
Considering the first one, this is true for:\\
$$
(\vect{a}\myplus\vect{b})\mytimes\vect{c}=(\vect{a}_0^{i-1}+1+\vect{a}_i^{s(\vect{a})})\mytimes_i\vect{c}=
\vect{a}_0^{i-1}+\vect{c}+\vect{a}_i^{s(\vect{a})}+\vect{b}=(\vect{a}\mytimes_i\vect{c})\myplus\vect{b}
$$
The other case corresponding to $i>s(\vect{a})$ is completely similar.\\
\item[$\product$] The case corresponding to the units is completely trivial.
\end{itemize}
\qed
\end{proof}
% hD
\subsection{The Hypersequent Calculus hD{}}\label{subsecthD}
We will now consider the \emph{string-based} hypersequent syntax from \cite{mfv:iwcs07}.
The reason for using the prefix \emph{hyper}
in the term \emph{sequent} is that the data-structure used in hypersequent antecedents is quite nonstandard.
% segtyp
A fundamental tool to build the data-structure of a sequent calculus for \D{}  is the notion of \emph{type-segment\/}.
For any type of sort $0$ $seg(A)=\{A\}$. If $s(A)>0$ then $seg(A)=\{\oldseg{0}{A},\cdots,\oldseg{s(A)}{A}\}$.
We call $seg(A)$ the set of type-segments of $A$. If $\mathbf C$ is a set of connectives, we can now define the set of type-segments corresponding to the set $\tp[\mathbf C]$ of types generated by the connectives $\mathbf C$ as $\mathbf{seg}[\mathbf C]=\bigcup_{A\in\mathbf{Tp}[\mathbf C]}$seg(A).
% hyperconfig
\noindent Type segments of sort $0$ are types. But, type segments of sort greater than
$0$ are no longer types. Strings of type segments can form meaningful logical material like
the set of configurations, which we now define. Where $\mathbf C$ is a set of connectives the {\em configurations\/}
$\config[\mathbf C]$ are defined in BNF unambiguously by
mutual recursion as follows, where $\Lambda$ is the empty string and \sep{} is the metalinguistic separator:
\disp{$
\footnotesize{
\begin{array}[t]{rcll}
\config[\mathbf C] & ::= & \Lambda \\%& A\ \mbox{for\ } S(A)=0\\
%\config & ::= & \sep\\
\config[\mathbf C] & ::= & \sep, \config[\mathbf C]\\
\config[\mathbf C] & ::= & A, \config[\mathbf C]\quad\mbox{for\ } s(A)=0\\
\config[\mathbf C]& ::= & \oldseg{0}{A},\config[\mathbf C], \oldseg{1}{A}, \cdots, \oldseg{s(A)-1}{A}, \config[\mathbf C],
\oldseg{s(A)}{A}_{s(A)}, \config[\mathbf C]\quad\mbox{for\ } s(A)>0
\end{array}}
\label{defhyperconfig}
$}
\noindent The intuitive semantic interpretation of the last clause from~(\ref{defhyperconfig}) consists of elements $\alpha_0+\beta_1+\alpha_1+\cdots+$ $\alpha_{n{-}1}+\beta_n+\alpha_n$
where $\alpha_0+\sep+\alpha_1+\cdots+\alpha_{n{-}1}+\sep+\alpha_n\! \in \oscott A\cscott$ and
 $\beta_1,\cdots, \beta_n$ are the interpretations of the intercalated configurations. 
 
 If the context is clear
 we will write $\config$ for $\config[\mathbf C]$, and likewise $\tp$, and $\mathbf{seg}$.
 
The syntax in which $\config$ has been defined is called \techterm{string-based hypersequent syntax}. An equivalent syntax for $\config$ is called \techterm{tree-based hypersequent syntax},
which was defined in \cite{mv:disp}, \cite{mvf:tdc}. For proof-search and human readability, the tree-based notation is more
convenient than the string-based notation, but for semantic purposes, the string-based notation turns out to be very useful
since the canonical model construction considered in Section~\ref{sectstrongcompldstarimpl} relies on the set of type-segments. %$(seg[\rightarrow]\cup\{\sep\}$).

In string-based notation the \techterm{figure} $\vect{A}$ of a type $A$ is defined as follows:
\disp{$
\vect{A} = \left\{
\begin{array}{ll}
A & \mbox{if\ } s(A)=0\\
\oldseg{0}{A},\sep,\oldseg{1}{A},\cdots,\oldseg{s(A)-1}{A},\sep,\oldseg{s(A)}{A} & \mbox{if\ } s(A)>0
\end{array}\right.$
}
The sort of a configuration is the number of metalinguistic separators it contains. We have $\config=
\bigcup_{i\geq 0}\config_i$, where $\config_i$ is the set of configurations of sort $i$. We define a more general notion
of configuration, namely preconfiguration. If $V$ denotes $\mathbf{seg}[\mathbf C]\cup \{\sep\}$, a preconfiguration $\Delta$ is simply a word of $V^*$. Obviously, we have that $\config\subsetneq V^*$. A preconfiguration $\Delta$ is proper iff 
$\Delta\not \in\config$. As in the case of configurations, preconfigurations have a sort.

Where $\Gamma$ and $\Phi$ are configurations and the sort of $\Gamma$ is
at least $1$,
$\Gamma|_k\Phi$ ($k>0$) signifies the configuration which is the result of
replacing the $k$-th separator in $\Gamma$ by $\Phi$. The notation
$\Delta\langle\Gamma\rangle$, which we call a configuration with a distinguished configuration $\Gamma$ abbreviates
the following configuration: $\Delta_0,\Gamma_0,\Delta_1,\cdots,\Delta_{s(\Gamma)},\Gamma_{s(\Gamma)},$ $
\Delta_{s(\Gamma)+1}$, where $\Delta_i\in\config$ but $\Delta_0$ and $\Delta_{s(\Gamma)+1}$ are possibly proper preconfigurations. 
 When a type-occurrence $A$ in a configuration is written without vectorial notation, that means that the
sort of $A$ is $0$. However, when one writes the metanotation for configurations $\Delta\langle\vect A\rangle$, this does not mean that the sort of $A$ is necessarily greater than $0$.

A \techterm{hypersequent\/} $\Gamma\yields A$ comprises an  
antecedent configuration in string-based notation of sort $i$ and a succedent type $A$ of sort $i$.
The hypersequent calculus for \D{} is as shown in Figure~\ref{Dseqcalc}.
% Insertion hD rules
\begin{figure}[h]
$$
\begin{array}{l}
\vect A\yields A\mbox{ if $A$ is primitive}\vtab
% Continuous unit
\prooftree
\Delta\langle\Lambda\rangle \yields A
\justifies
\Delta\langle I\rangle\yields A
\using IL
\endprooftree
\tb
\prooftree
\justifies
\Lambda\yields I
\using IR
\endprooftree\vtab
% Discontinuous unit
\prooftree
\Delta\langle\sep\rangle \yields A
\justifies
\Delta\langle \vect J\rangle\yields A
\using JL
\endprooftree
\tb
\prooftree
\justifies
\sep\yields J
\using JR
\endprooftree\vtab
\prooftree
\Gamma\yields A
\tb
\Delta\langle\vect{B}\rangle\yields C
\justifies
\Delta\langle\vect{B/A},\Gamma\rangle\yields C
\using/ L
\endprooftree
\tb
\prooftree
\Delta,\vect A \yields B
\justifies
\Delta\yields B/A
\using/ R
\endprooftree\vtab
% \bsl rules
\prooftree
\Gamma\yields A
\tb
\Delta\langle\vect{B}\rangle\yields C
\justifies
\Delta\langle\Gamma,\vect{A\bsl B}\rangle\yields C
\using\bsl L
\endprooftree
\tb
\prooftree
\vect A,\Delta \yields B
\justifies
\Delta\yields A\bsl B
\using\bsl R
\endprooftree\vtab
% Continuous product
\prooftree
\Delta\langle\vect{A},\vect{B}\rangle
\yields  C
\justifies
\Delta\langle\vect{A\bullet B}\rangle
\yields  C
\using\bullet L
\endprooftree
\tb
\prooftree
\Delta\yields A
\tb
\Gamma\yields  B
\justifies
\Delta,\Gamma\yields A\bullet B
\using\bullet R
\endprooftree\vtab
% Implicative Discontinuous
\prooftree
\Gamma\yields A
\tb
\Delta\langle\vect{B}\rangle\yields C
\justifies
\Delta\langle\vect{B\extract_i A}|_i\Gamma\rangle\yields C
\using\extract_i L
\endprooftree
\tb
\prooftree
\Delta|_i\vect A\yields B
\justifies
\Delta\yields B\extract_i A
\using\extract_i R
\endprooftree\vtab
% infix rules
\prooftree
\Gamma\yields A
\tb
\Delta\langle\vect{B}\rangle\yields C
\justifies
\Delta\langle\Gamma|_i\vect{A\infix_i B}\rangle\yields C
\using\infix_i L
\endprooftree
\tb
\prooftree
\vect A|_i\Delta\yields B
\justifies
\Delta\yields A\infix_i B
\using\infix_i R
\endprooftree\vtab
% Discontinuous product rules
\prooftree
\Delta\langle\vect{A}|_i\vect{B}\rangle\yields C
\justifies
\Delta\langle\vect{A\odot_i B}\rangle\yields C
\using
\odot_i L
\endprooftree
\tb
\prooftree
\Delta\yields A\tb \Gamma\yields B
\justifies
\Delta|_i\Gamma\yields A\odot_i B
\using
\odot_i R
\endprooftree
\end{array}
$$
\caption{Hypersequent Calculus for D}
\label{Dseqcalc}
\end{figure}
The following lemma is useful for the strong completeness results of section~\ref{sectstrongcompldstarimpl}:
\begin{lemma}
Recall that \config{} is a subset of $V^*=(\mathbf{seg}[C]\cup\{1\})^*$. We have that:
\begin{itemize}
\item[i)] \config{} is closed by concatenation and intercalation.
\item[ii)] If $\Delta\in V^*$, $\Gamma\in\config{}$, and $\Delta,\Gamma\in\config$,
then $\Delta\in\config$. Similarly, if we have $\Gamma,\Delta\in\config$ instead of $\Delta,\Gamma\in\config$.
Finally, If $\Delta\in V^*$, $\Gamma\in\config{}$, and $\Delta |_i\Gamma\in\config$,
then $\Delta\in\config$.
\end{itemize}
\label{propshconfigs}
\end{lemma}
\begin{proof}
Propositions (i) and ii) are both proved via the BNF derivations of (\ref{defhyperconfig}). The details
of the proof are rather tedious but not difficult.
\end{proof}

What is the connection between the calculi $\cD$ and $\hD$? In \cite{valentin:phd} a (faithful) embedding
translation is proved. Let $\Delta$ denote a configuration. We define its \emph{type-equivalent\/} $\Delta^\bullet$,
which is a type which has the same algebraic meaning as $\Delta$. Via the BNF formulation of $\mathcal{O}[\mathbf C]$ in (\ref{defhyperconfig}) one defines recursively $\Delta^\bullet$ as follows:
$$
\begin{array}{lll}
\Lambda^\bullet\bydef I\\
(1,\Gamma)^\bullet\bydef\mathbf J\bullet\Gamma^\bullet\\
(A,\Gamma)^\bullet\bydef A\bullet \Gamma^\bullet\mbox{, if $s(A)=0$}\\
(\oldseg{0}{A},\Delta_1, \cdots, \oldseg{s(A)-1}{A}, \Delta_{s(A)},\oldseg{s(A)}{A}_{s(A)},\Delta_{s(A+1)})^\bullet\bydef\\%&\bydef&
((\cdots(A\odot_1 \Delta_1^\bullet)\cdots)\odot_{1+s(\Delta_1)+\cdots+s(\Delta_{s(A)}}\Delta_{s(A)}^\bullet)\bullet
\Delta_{s(A)+1}^\bullet\mbox{, if $s(A)>0$}
\end{array}
$$
The semantic interpretation of a configuration $\Delta$ (for a given valuation $v$)
is $\widehat v(\Delta)\bydef\widehat v(\Delta^\bullet)$.
The embedding translation is as follows. For any $\Delta\in\mathcal O, \cD\vdash \Delta^\bullet\rightarrow A$ iff
$\hD\vdash \Delta\yields A$. 

% End insertion D rules

% (S,\Sigma)-algebras
%The \emph{displacement algebra} defined by a syntactical algebra %$(L, +, 0, \sep)$
%is the $\omg$-sorted algebra with the $\omg$-sorted signature $\Sigma_D=(\conc, \{\interc{i}\}_{i>0},\zero,\sep)$ with sort 
%functionality $((i,j\rightarrow i+j)_{i,j\geq 0},(i,j\rightarrow i+j-1)_{i>0,j\geq 0},0,1)$:
% The (hyper)sequent calculus and its natural standard syntactic interpretation
%\newpage
\subsection{Some special DAs}\label{subsectspecialdas}
The standard DA $\alg{S}$, induced by the separated monoid with generator set $V=\segtypstar\cup\{1\}$, plays an important
role. The interpretation of the signature \sigmad{} in $\car{S}$ is:
\disp{
$\alg{S}=(V^*,+,\{|_i\}_{i>0},\Lambda,\sep)$
\label{algS}
}
Here, $+$ denotes concatenation, and $\{|_i\}_{i>0}$ i-th intercalation. We have seen in Section~\ref{sectdda} that \config{} is closed by concatenation 
$+$ and intercalation $|_i$, $i>0$, i.e.\ $\mathcal{C}=(\config, +,(|_i)_{i>0},\Lambda,1)$ is a $\sigmad$-subalgebra
of the standard DA $\alg S$.
%\footnote{Observe
%that the sort functionalities of $(,)$ and $|_i$ ($i>0$) are respectively $(i,j\rightarrow i+j)_{i,j\geq 0}$, and 
%$(i,j\rightarrow i+j-1)_{i>0,j\geq 0}$, where $s(0)=0,\mbox{ and }, s(1)=1 $.}
Since $\da$ is a variety,\footnote{We recall that varieties are closed by subalgebras, homomorphic images, and 
direct products.} $\alg{C}$ is a (general) DA, concretely a nonstandard DA. To see that $\alg{C}$ cannot 
be standard we notice that the sort domains of 
$\alg{C}$ are not separated by $\{1\}$. Recall that $\car{C}=\displaystyle\bigcup_{i\geq 0}\config_i$ ($\car{C}_i=\config_i$,
for every $i\geq 0$). We have that:
\disp{$\mbox{For $i>0$,}\car{C}_i\neq\underbrace{\config_0\comp{}\cdots\comp{}\config_0}_{\mbox{$i$ times}}$
\label{disp1}
}
Because, for example, let us take $\vect{p\extract_1p}=\oldseg{0}{p\extract_1 p},1,\oldseg{1}{p\extract_1 p}$, where $p\in\pr_0$. The type
$p\extract_1 p$ has sort $1$, but clearly neither $\oldseg{0}{p\extract_1 p}$ nor $\oldseg{1}{p\extract_1 p}$ are members of $\config_0$. In fact, we have the proper inclusion:
\disp{$\mbox{For $i>0$,}\underbrace{\config_0\comp{}\cdots\comp{}\config_0}_{\mbox{$i$ times}}\subsetneq\car{C}_i$
\label{disp2}
}
It follows that the class of standard DAs is a \emph{proper\/} subclass of the class of general DAs.
%The Lindenbaum algebra $\alg{L}$ defined in the previous Subsection is a nonstandard DA, because again, the sort domains are not
%separated by $\{1\}$. The initial $\sigmad$-algebra $\alg{N}$ in $\da$ is the standard displacement algebra induced by the singleton
%generator set $\{1\}$, where $\alg{N}=(\omg,+^{\cal{N}}, \{\times_i^{\cal{N}}\}_{i>0},0^{\cal{N}},\sep^{\cal{N}})$. The
%elements of the signature are interpreted as follows:
%\disp{
%$
%\begin{array}{lll}
%0^{\cal{N}}&\bydef&0\\
%1^{\cal{N}}&\bydef&1\\
%+^{\cal{N}}(n,m)&\bydef&n+m\mbox{, where $n,m\in\omg$}\\
%\mbox{Where $i>0$, } \times_i^{\cal{N}}(n,m)&\bydef&n+m-1\mbox{, where $n>0$ and $m\geq 0$}
%\end{array}
%\label{interpsignnumalg}
%$
%\label{numdaN}
%}

\subsection{Synthetic Connectives and the Implicative fragment}

\label{subsectdstarimp}

From  a logical point of view, synthetic connectives abbreviate formulas in sequent systems. They form new connectives with left and right sequent rules. Using a linear logic
slogan, synthetic connectives help to eliminate some \techterm{bureaucracy} in Cut-free proofs and in the (syntactic) 
Cut-elimination algorithms (see \cite{valentin:phd}). We consider here a set of synthetic connectives which are of linguistic interest:
\begin{itemize}
\item The binary non-deterministic implications $\nextract$, and $\ninfix$.
\item The unary connectives $\leftproj$, $\rightproj$ and $(\splitk k)_{k>0}$, which are called respectively
left projection, right projection, and split.
\end{itemize}
Together with the binary deterministic implications $\bsl$, /, $(\extract_i)_{i>0}$, $(\infix_i)_{i>0}$,
these constitute what we call
\emph{implicative\/} connectives.
These connectives are incorporated in the recursive definitions of \tp{}, $\mathbf{seg}$, and $\mathcal O$.
We denote this implicative fragment as $\D[\rightarrow]$. We write also $\tp[\rightarrow]$, $\mathbf{seg}[\rightarrow]$,
and $\mathcal O[\rightarrow]$, although, as usual, if the context is clear we will avoid writing $[\rightarrow]$. The 
intuitive semantic interpretation of the {implicative} connectives can be found in Figure~\ref{synconndefs}.
Figure~\ref{chapt4hDrulesdetconns} and Figure~\ref{chapt4hDrulesnondetconns} correspond to their hypersequent rules.

Besides the usual continuous and discontinuous implications, the nondeterministic discontinous implications are used to
account for particle shift nondeterminism where the object can be intercalated between the verb and the particle,
or after the particle. For a particle verb like \emph{call$+\sep+$up\/} we can give the lexical assigment $\leftproj(\splitk 1(N\bsl S)\nextract N)$. Projections can be used to account for the cross-serial dependencies of Dutch. The split connective
can be used for parentheticals like \emph{fortunately\/} with the type assignment $\splitk 1 S\infix_1 S$.

% Insertion D with defined connectives
%\disp{
\begin{figure}[h]
\begin{center}
$
\begin{array}{rcll}
%\mbox{The deterministic connectives}\\
\oscott\leftproj A\cscott & \bydef & \oscott A\cscott\myover \mydiscunit & \mbox{left projection}\\
\oscott\rightproj A\cscott & \bydef & \mydiscunit\myunder \oscott A\cscott & \mbox{right projection}\\
\oscott\splitk{i} A\cscott & \bydef & \oscott A\cscott\extract_i \myconunit & \mbox{$i$-th split}\\
\oscott B\nextract A\cscott & \bydef &\oscott B\extract_1 A\cscott\,\cap\,\cdots\,\cap\,\oscott B\extract_{s(B)-s(A)+1} A\cscott & \mbox{nondeterministic extract}\\
\oscott A\ninfix B\cscott & \bydef & \oscott A\infix_1 B\cscott\,\cap\,\cdots\,\cap\,\oscott A\infix_{s(B)-s(A)+1} B\cscott & \mbox{nondeterministic infix}\\
\end{array}
$
\end{center}
\caption{Semantic interpretation in standard DAs for the set of synthetic connectives}
\label{synconndefs}
\end{figure}
%}
% End insertion D with defined connectives
% Formulation hypersequent rules for defined connectives
\begin{figure}[h]
$
\prooftree
\Gamma\langle\vect A\rangle\yields B
\justifies
\Gamma\langle\vect{\leftproj A},\sep\rangle\yields B
\using\leftproj L
\endprooftree
\tb
\prooftree
\Gamma,\sep\yields A
\justifies
\Gamma\yields \leftproj A
\using\leftproj R
\endprooftree
$
\vtab
$
\prooftree
\Gamma\langle\vect A\rangle\yields B
\justifies
\Gamma\langle\sep,\vect{\rightproj A}\rangle\yields B
\using\rightproj L
\endprooftree
\tb
\prooftree
\sep,\Gamma\yields A
\justifies
\Gamma\yields \rightproj A
\using\rightproj R
\endprooftree
$
\vtab
$\prooftree
\Delta\langle\vect{B}\rangle\yields C
\justifies
\Delta\langle\vect{\splitk{i}B}|_i\Lambda\rangle\yields C
\using\splitk{i} L
\endprooftree
\tb
\prooftree
\Delta|_i \Lambda\yields B
\justifies
\Delta\yields \splitk{i}B
\using \mysplit{}R
\endprooftree
$
\caption{Hypersequent rules for synthetic unary implicative connectives}%deterministic synthetic connectives}
\label{chapt4hDrulesdetconns}
\end{figure}
% Nondeterministic synthetic rules
\begin{figure}
$
\prooftree
\Delta\yields A
\tb
\Gamma\langle\vect B\rangle \yields C
\justifies
\Gamma\langle\vect{ B\nextract A}|_i\Gamma\rangle \yields C
\using\nextract L
\endprooftree
\tb
\prooftree
\Delta|_1\vect A\yields B
\tb
\cdots
\tb
\Delta|_d\vect A\yields B
\justifies
\Delta\yields B\nextract A
\using\nextract R
\endprooftree
$
\vtab
$
\prooftree
\Delta\yields A
\tb
\Gamma\langle\vect B\rangle \yields C
\justifies
\Gamma\langle\Gamma|_i\vect{ A\ninfix B}\rangle \yields C
\using\ninfix L
\endprooftree
\tb
\prooftree
\vect A|_1\Delta\yields B
\tb
\cdots
\tb
\vect A|_a\Delta\yields B
\justifies
\Delta\yields A\ninfix B
\using\ninfix R
\endprooftree
$
\caption{Hypersequent calculus rules for nondeterministic synthetic connectives}
\label{chapt4hDrulesnondetconns}
\end{figure}

%
%\section{Strong Completeness w.r.t.\ Powerset Residuated Displacement Algebras over Standard Displacement Algebras}
\section{Strong Completeness of the implicative fragment w.r.t.\ L-models}\label{sectstrongcompldstarimpl}
In this section we prove two strong completeness theorems in relation to the
implicative fragment.
In order to prove them, we demonstrate first strong completeness of $\hD[\rightarrow]$ w.r.t.\ powerset residuated DAs over standard DAs with a countable set of generators.
% Strong completeness of \Dstartimp w.r.t. the class \prstd{}

Let $V=\segtypstar[\rightarrow]\cup\{\sep\}$. Clearly, $V$ is countably infinite since $\segtypstar[\rightarrow]$ is the countable union 
$\displaystyle \bigcup_i \segtypstar[\rightarrow]_i$, where each
$\segtypstar[\rightarrow]_i$ is also countably infinite. Let us consider the standard DA $\alg{S}$
(from~(\ref{algS})), induced by the (countably) infinite set of generators $V$:  
$$\alg{S}=(V^*,+,\{|_i\}_{i>0},\Lambda,\sep)$$
We define some notation:
\begin{definition}
For any type $C\in\tp[\rightarrow]$ and set $R$ of non-logical axioms:
$$
\begin{array}{lll}
\invprov{C}_R&\bydef&\{\Delta:\Delta\in\config\mbox{ and }R\vdash \Delta\yields C\}
\end{array}
$$
%\label{definvprov}
\end{definition}
In practice, when the set of hypersequents $R$ is clear from the context, we simply write $\invprov C$ instead
of $\invprov{C}_R$. 
\begin{lemma}\rm{(\textbf{Truth Lemma})}\label{truthlemma}\\\\
Let $\powerset{S}$ be the powerset residuated DA over the standard DA $\alg{S}$ from~(\ref{algS}). 
Let $v_R$ be the following valuation on the powerset $\powerset{S}$:
$$
\begin{array}{lll}
\mbox{For every }p\in\pr,\mbox{ }v_R(p)&\bydef&\invprov{p}_R
\end{array}
$$
Let $\model{M}=(\powerset{S},v_R)$ be called as usual the canonical model. 
The following equality holds:
$$
\begin{array}{lll}
\mbox{For every }C\in\tp[\rightarrow],\mbox{ }{\widehat{v}}_R(C)&=&\invprov{C}_R
\end{array}
$$
\end{lemma}
\begin{proof}
We proceed by induction on the structure of type $C$; we will write $\widehat v$ instead of $\widehat{v_R}$, and $\invprov{\cdot}$ instead of $\invprov{\cdot}_R$;
we will say that an element $\Delta\in\widehat v(A)$ is
\emph{correct}\footnote{Recall that a priori $\Delta\in\car S$, which is equal to 
$(\segtyp[\rightarrow]\cup\{\sep\})^*$.} iff $\Delta\in\config[\mathbf C]$.
\begin{itemize}
\item[$\product$] $C$ is primitive. True by definition.\\
\item[$\product$] $C=B\extract_i A$. Let us see:
$$\invprov{B\extract_i A}\subseteq \widehat v(B\extract_i A)$$
Let $\Delta$ be such that $R\vdash \Delta\yields B\extract_i A$. 
Let $\Gamma_A\in\widehat v (A)$. By induction hypothesis (i.h.), $\widehat v (A)=\invprov{A}$.
Hence, $R\vdash\Gamma_A\yields A$
We have:
$$
\prooftree
\Delta\yields B\extract_i A
\tb
\vect{B\extract_i A}|_i \Gamma_A\yields B
\justifies
\Delta |_i\Gamma_A\yields B
\using Cut
\endprooftree
$$
By i.h., $\widehat v(B)=\invprov{B}$. It follows that $\Delta |_i\Gamma_A\in\widehat v(B)$, hence 
$\Delta\in\widehat v(B\extract_i A)$. Whence, $\invprov{B\extract_i A}\subseteq \widehat v(B\extract_i A)$.\\

Conversely, let us see:
$$\widehat v(B\extract_i A)\subseteq\invprov{B\extract_i A}$$
Let $\Delta\in\widehat v(B\extract_i A)$. By i.h. $\widehat v(A)=\invprov{A}$. For any type $A$, we have eta-expansion,
i.e.\ $\vect A\yields A$\footnote{By simple induction on the structure of types.}. Hence, $\vect A\in\widehat v(A)$. We have that 
$\Delta |_i\vect A\in\widehat v(B)$. By i.h., $\Delta |_i\vect A\yields B$. Since $\vect A$ is correct,
and by i.h.\ $\Delta |_i\vect A$ is correct, by Lemma~\ref{propshconfigs}, $\Delta$ is correct. By applying the $\extract_i$ right rule to the provable hypersequent $\Delta |_i\vect A\yields B$ we get:
$$\Delta\yields B\extract_i A$$ 
This ends the case of $B\extract_i A$.\\
\item[$\product$] $C= A\infix_i B$. Completely similar to case $B\extract_i A$.\\
\item[$\product$] $C= B/A$ or $A\bsl B$. Similar to the disconcontuous case.
\commentout{
Let us see:
$$\invprov{B/ A}\subseteq \widehat v(B/A)$$
Let $\Delta$ be such that $R\vdash \Delta\yields B/A$. 
Let $\Gamma_A\in\widehat v(A)$. By i.h., $\widehat v(A)=\invprov{A}$.
Hence, $R\vdash\Gamma_A\yields A$
We have:
$$
\prooftree
\Delta\yields B/A
\tb
\vect{B/A},\Gamma_A\yields B
\justifies
\Delta, \Gamma_A\yields B
\using Cut
\endprooftree
$$
By i.h., $\widehat v(B)=\invprov{B}$. It follows that $\Delta, \Gamma_A\in\widehat v(B)$. Whence, 
$\invprov{B/A}\subseteq \widehat v(B/A)$.\\

Conversely, let us see:
$$\widehat v(B/A)\subseteq\invprov{B/ A}$$
Let $\Delta\in\widehat v(B/A)$. By i.h. $\widehat v(A)=\invprov{A}$. For any type $A$, we have eta-expansion,
i.e.\ $\vect A\yields A$. Hence, $\vect A\in\widehat v(A)$. We have that 
$\Delta,\vect A\in\widehat v(B)$. By i.h., $\Delta,\vect A\yields B$. Since $\vect A$ is correct,
and by i.h.\ $\Delta,\vect A$ is correct, by Lemma~\ref{propshconfigs} $\Delta$ is correct. By applying the $/$ right rule to the provable hypersequent $\Delta,\vect A\yields B$ we get:
$$\Delta\yields B/A$$ 
This ends the case of $B/A$.\\
\item[$\product$] $C= A\bsl B$. Completely similar to the case $C=B/A$.\\
}
% Nondeterministic connectives
\item[$\product$] Nondeterministic connectives. Consider the case $C= B{{\Uparrow}} A$.\\
$$\invprov{B{{\Uparrow}}A}\subseteq \widehat v(B{{\Uparrow}} A)$$
Let $\Gamma_A\in \widehat v(A)$. By i.h, $\Gamma_A\yields A$. Let $\Delta\yields B{\Uparrow} A$. By $s(B)-s(A)+1$ applications of ${\Uparrow}$ left rule, we have
$$
\prooftree
\Gamma_A\yields A
\tb 
\vect B\yields B\mbox{, by eta-expansion}
\justifies
\vect{B{\Uparrow} A}|_i\Gamma_A\yields B, \mbox{ for }i=1,\cdots,s(B)-s(A)+1
\using{\Uparrow} L
\endprooftree
$$
By $s(B)-s(A)+1$ Cut applications with $\Delta\yields B{\Uparrow} A$, we get:
$$\Delta |_i\Gamma_A\yields B
$$
Hence, for $i=1,\cdots s(B)-s(A)+1$, by i.h.\ $\Delta |_i\Gamma_A\in\widehat v(B)$.
Hence, $\Delta\in\widehat v(B{{{\Uparrow}}} A)$.\\

Conversely, let us see:
$$\widehat v(B{\Uparrow} A)\subseteq \invprov{B{\Uparrow} A}$$
By i.h, we see that $\vect A\in\widehat v(A)$. Let $\Delta\in\widehat v(B{\Uparrow} A)$. This means that
for every $i=1,\cdots,s(B)-s(A)+1$ $\Delta |_i\vect A\in\widehat v(B)$. By i.h., 
$\Delta|_i\vect A\yields B$. By a similar reasoning to the deterministic case $C=B\extract_i A$, we see
that $\Delta$ is correct. We have that:
$$
\prooftree
\Delta|_1\vect A\yields B
\tb
\cdots
\tb
\Delta|_{s(B)-s(A)+1}\vect A\yields B
\justifies
\Delta\yields B{\Uparrow} A
\using{\Uparrow} R
\endprooftree
$$
\item[$\product$] $C=A{\Downarrow} B$ is completely similar to the previous one.\\

%Let us see the cases corresponding to the unary (implicative) connectives.
\item[$\product$] $C=\leftproj A$. Let us see:
$$[\leftproj A]\subseteq \widehat v(\leftproj A)$$
Let $\Delta\in[\leftproj A]$. Hence, $\Delta\yields\leftproj A$. We have that:
$$
\prooftree
\Delta\yields\leftproj A
\tb
\prooftree
\vect A\yields A
\justifies
\vect{\leftproj A},1\yields A
\using\leftproj L
\endprooftree
\justifies
\Delta,1\yields A
\using Cut
\endprooftree
$$
By i.h., $\Delta,1\in\widehat v(A)$. Hence, $\Delta\in\widehat v(\leftproj A)$.\\

Conversely, let us see:
$$\widehat v(\leftproj A)\subseteq[\leftproj A]$$
Let $\Delta\in\widehat v(\leftproj A)$. By definition, $\Delta,1\in\widehat v(A)$. By i.h., $\Delta,1\yields A$, and 
by lemma~\ref{propshconfigs}, $\Delta$ is correct. By application of $\leftproj$ right rule, we get:
$$\Delta\yields\leftproj A$$
This proves the converse.\\
\item[$\product$] $C=\rightproj A$ is completely similar to the previous one.\\
\item[$\product$] $C= \splitk{k}\!A$. Let us see:
$$\invprov{\splitk{k}\!A}\subseteq\widehat v(\splitk{k}\!A)$$
Let $\Delta\yields\splitk{k}\!A$. We have that:
$$
\prooftree
\Delta\yields\splitk i\! A
\tb
\prooftree
\vect A\yields A
\justifies
\vect{\splitk i A}|_k\Lambda\yields A
\using\splitk k L
\endprooftree
\justifies
\Delta|_k\Lambda\yields A
\using Cut
\endprooftree
$$
By i.h., $\Delta\in\widehat v(\splitk k A)$.\\

Conversely, let us see that:
$$\widehat v(\splitk{k}\!A)\subseteq\invprov{\splitk{k}\!A}$$
Let $\Delta\in\widehat v(\splitk{k}\!A)$. By definition, $\Delta|_k\Lambda\in\widehat v(A)$. By i.h. and lemma~\ref{propshconfigs},
$\Delta$ is correct and $\Delta|_k\Lambda\yields A$. By application of the $\splitk k$ right rule:
$$\Delta\yields \splitk k A$$
Hence, $\Delta\in\invprov{\splitk k A}$.\qed
\end{itemize}
% Identity lema
\end{proof}
By induction on the structure of \config{}, see~(\ref{defhyperconfig}), one proves the following lemma:
\begin{lemma}\rm{(\textbf{Identity lemma})}\label{identlemma}\\\\
For any $\Delta\in\config$, $\Delta\in\widehat v(\Delta)$.
\end{lemma}

Let $(A_i)_{i=1,\cdots,n}$ be the sequence of type-occurrences in a configuration $\Delta$. Let 
$\Delta\subst{\Gamma_1\cdots \Gamma_n}{A_1\cdots A_n}$ be the result of replacing every type-occurrence
$A_i$ with $\Gamma_i$.
Recall that we have fixed a set of hypersequents $R$. We have the lemma:
\begin{lemma}
$\model{M}=(\powerset{S},v)\models R$
\label{truthlemma2}
\end{lemma}
\begin{proof}
Let $(\Delta\yields A)\in R$. For every type-occurrence $A_i$ in $\Delta$ (we suppose that the sequence of type
occurrences in $\Delta$ is $(A_i)_{i=1,\cdots,n}$), we have by the Truth Lemma that
$\widehat v(A_i)=\invprov{A_i}_R$. For any $\Gamma_i\in\widehat v(A_i)$, we have by the 
Truth Lemma that $R\vdash \Gamma_i\yields A_i$. Since $(\Delta\yields A)\in R$, we have then that
$R\vdash\Delta\yields A$. By $n$ applications of the Cut rule with the premises $\Gamma_i$ we get from
$R\vdash\Delta\yields A$ that $R\vdash\Delta\subst{\Gamma_1\cdots \Gamma_n}{A_1\cdots A_n}\yields A$.
We have that $\widehat v(\Delta)=\{\Delta\subst{\Gamma_1\cdots \Gamma_n}{A_1\cdots A_n}:\Gamma_i\in
\widehat v(A_i)\}$. Since, we have $R\vdash\Delta\subst{\Gamma_1\cdots \Gamma_n}{A_1\cdots A_n}\yields A$,
again by the Truth Lemma, $\Delta\subst{\Gamma_1\cdots \Gamma_n}{A_1\cdots A_n}\in\widehat v(A)$. 
We have then that $\widehat v(\Delta)\subseteq \widehat v(A)$. We are done.\qed
\end{proof}
% Strong completeness of \Dstarimp w.r.t. \prsd{}
\begin{theorem}
\Dstarimp{} is strongly complete w.r.t.\ the class \prsd{}.
\label{compldstarimprstd}
\end{theorem}
\begin{proof}
Suppose $\prsd( R )\models\Delta\yields A$. Hence, in particular this is true of the canonical
model $\model{M}$. Since $\Delta\in\widehat v(\Delta)$, it follows that $\Delta\in\widehat v(A)$.
By the Truth Lemma, $\widehat v(A)=\invprov{A}_R$. Hence, $R\vdash\Delta\yields A$. We are done.\qed
\end{proof}
% Strong completeness of \Dstartimp without \splitk{k} w.r.t. L-models
% Finitely generated standard DAs
% XXX
% Coding monomorphism

We shall also prove strong completeness w.r.t.\ L-models over the set of connectives $\Sigma[\rightarrow]-\mathbf{split}$,
where $\mathbf{split}=\{\splitk{k}:k>0\}$.
Since the canonical model $\cal S$ is countably infinite, $\car{S}$ is in bijection with a set $V_1=(a_i)_{i>0}\cup\{1\}$
via a mapping $\Phi$. Let $\cal A$ be the standard DA associated to $V_1$. ${\Phi}$ extends to an isomorphism of standard DAs between $\cal S$ and $\cal A$, and then induces an isomorphism $\bar{\Phi}$ of residuated powerset DAs over standard DAs.
Let $\alg{B}$ be a standard DA generated by the finite set of generators $V_2=\{a,b,1\}$. We have that $\car{A}=V_1^*$, and
$\car{B}=V_2^*$.
Let $\rho$ be the following injective mapping from $V_1$ into $V_2^*$:
\disp{
$
\begin{array}[t]{lllllll}
\rho(1)&=&1&\quad&
\rho(a_i)&=& a+b^i+a
\end{array}
$
\label{rhomap}}
The mapping $\rho$ extends recursively to the morphism of standard DAs.\commentout{
$$
\begin{array}{llll}
\rho:& \alg{A}&\longrightarrow&\alg{B}\\
& 0&\mapsto& 0\\
& 1&\mapsto& 1\\
& \rho(w_1+w_2)&\mapsto& \rho(w_1)+\rho(w_2)\\
& \rho(w_1\times_i w_2)&\mapsto& \rho(w_1)\times_i\rho(w_2)
\end{array}
$$}%\label{rhomorph}}
Clearly
$\rho$ is injective by freeness\footnote{Since the underlying structures are free monoids we can apply left/right cancellation.} of the underlying free monoids $\car{A}$ and $\car{B}$. 
The mapping $\rho$ is a monomorphism of DAs which induces
a monomorphism $\bar{\rho}$ of residuated powerset DAs over DAs.
\noindent Let $A$, $B$ and $C$ range over subsets of $\car{A}$ such that they are non-empty and 
different from $\{\epsilon\}$. Since $\rho$ is injective, so is $\bar{\rho}$. The following equalities hold:
\disp{
$\begin{array}[t]{lllllllll}
\bar{\rho}(A\cdot B)&=&\bar{\rho}(A)\cdot\bar{\rho}(B)&\mbox{ }\bar{\rho}(A\comp{i} B)&=&\bar{\rho}(A)\comp{i} \bar{\rho}(B) & \bar{\rho}(A\myover \mydiscunit)&=&\bar{\rho}(A)\myover \bar{\rho}(\mydiscunit)\\
\bar{\rho}(A\myunder B)&=&\bar{\rho}(A)\myunder \bar{\rho}(B)&\mbox{ }\bar{\rho}(B\myover A)&=&\bar{\rho}(B)\myover \bar{\rho}(A)&\bar{\rho}(\mydiscunit\myunder A)&=&\mydiscunit\myunder \bar{\rho}( A )\\
\bar{\rho}(A\myinfix_i B)&=&\bar{\rho}(A)\myinfix_i \bar{\rho}(B)&\mbox{ }\bar{\rho}(B\myextract_i A)&=&\bar{\rho}(B)\myextract_i \bar{\rho}(A)\\
%\mbox{ }\bar{\rho}(B\myextract_i \myconunit)&=&\bar{\rho}(B)\myextract_i \myconunit
\end{array}
$
\label{rhoprops1}}% \bar{\rho} properties
The equalities
(\ref{rhoprops1}) are due to the fact that 1) $\bar{\rho}$ is a monomorphism of DAs, 2) we can apply cancellation, and
3) the subsets considered are non-empty and different from $\{\epsilon\}$. Since $\bar\rho$ is injective, 
arbitrary families of (same-sort) subsets satisfy $\bar\rho(\bigcap_{i\in I} X_i)=\bigcap_{i\in I}\bar\rho(X_i)$.
Moreover, using (\ref{rhoprops1}) one proves:
\disp{
$
\begin{array}[t]{lll}
\bar{\rho}(\displaystyle\bigcap_{i=1}^{s(B)-s(A)+1}B\myextract_i A) &=& \displaystyle\bigcap_{i=1}^{s(B)-s(A)+1}\bar{\rho}( B)\myextract_i\bar{\rho}( A)\\
\bar{\rho}(\displaystyle\bigcap_{i=1}^{s(B)-s(A)+1}A\myinfix_i B) &=& \displaystyle\bigcap_{i=1}^{s(B)-s(A)+1}\bar{\rho}( A)\myinfix_i\bar{\rho}( A)
\end{array}
$
\label{rhoprops2}
}
Recall that $v$ is the valuation of the canonical model $\powerset{S}$.
Consider the following composition of mappings: $\Pr\overset{v}{\longrightarrow}\powerset S
\overset{\bar\Phi}{\longrightarrow}\powerset A\overset{\bar\rho}{\longrightarrow}\powerset B$. Put
$w=\bar{\rho}\circ\bar{\Phi}\circ v$. We have that $\widehat{w}=\bar{\rho}\circ\bar{\Phi}\circ \widehat{v}$.
In order to prove the last equality we have to see that $\bar{\rho}\circ\bar{\Phi}\circ \widehat{v}$
is a monorphism of DAs.\footnote{There is a unique morphism of DAs extending $w$.}
For example, if $A$ and $B$ are types, one has:
$$
\begin{array}{l}
(\bar{\rho}\circ\bar{\Phi}\circ \widehat{v})(B\extract_{k} A) =\bar{\rho}(\bar{\Phi}(\widehat{v}(B\extract_k A))
=\bar{\rho}(\bar{\Phi}(\widehat{v}(B)\myextract_k \widehat v(A))=\\
\bar\rho(\bar\Phi(\widehat{v}(B))\myextract_k\widehat v(A))\mbox{, $\bar\Phi$ is an isomorphism of DAs} =\\
\bar\rho(\bar\Phi(\widehat{v}(B)))\myextract_k \bar\rho(\bar\Phi(\widehat{v}(A)))\mbox{, $\bar\rho$ satisfies
(\ref{rhoprops1})}
\end{array}
$$
Similar computations give the desired equalities for the remaining considered implicative connectives.\footnote{
Including also projection connectives.}
Given a set of non-logical axioms $R$, $R\vdash_{\mathbf{hD}}\Delta\yields A$ iff $\widehat v(\Delta)
\subseteq\widehat v (A)$ (we write $\widehat v$ instead of $\widehat{ v_R}$) iff $(\bar\Phi\circ \widehat v)
(\Delta)\subseteq (\bar\Phi\circ \widehat v)(A)$ iff $(\bar\rho\circ\bar\Phi\circ \widehat v)(\Delta)\subseteq
(\bar\rho\circ\bar\Phi\circ \widehat v)(A)$. We have proved:
\begin{theorem}
$\Dstar[\Sigma_{\rightarrow}{-}\mathbf{split}]$ is strongly complete w.r.t.\ L-models.
\label{compldstarimLmodels}
\end{theorem}
\begin{corollary}
$\Dstar[\Sigma_{\rightarrow}{-}\mathbf{split}]$ is strongly complete w.r.t.\ powerset residuated DAs overs standard DAs with $3$ generators.
\label{corolthreegens}
\end{corollary}

% end full completeness
\bibliographystyle{plain}
\bibliography{bib150311}

\end{document}